\documentclass[final]{article}

\usepackage{longtable}
\usepackage{multirow}
\usepackage{booktabs}
\usepackage{color}
\usepackage{graphicx}
\usepackage{amsmath}
\usepackage{amsfonts}

\newtheorem{theorem}{Theorem}[section]

\title{Modeling tumor cell heterogeneity and plasticity in adaptive therapy}

\author{Rui Yue\thanks{School of Mathematical Sciences, Center for Applied Mathematics, Tiangong University, Tianjin 300387, China}, Chenghang Li\thanks{School of Mathematical Sciences, Center for Applied Mathematics, Tiangong University, Tianjin 300387, China}, Jinzhi Lei\thanks{School of Mathematical Sciences, Center for Applied Mathematics, Tiangong University, Tianjin 300387 (Email: jzlei@tiangong.edu.cn).}}

\begin{document}
\maketitle

\begin{abstract}
Adaptive therapy (AT) is designed to postpone the emergence of drug resistance by exploiting evolutionary competition among tumor subclones. Most mathematical models of AT assume a binary population structure of drug-sensitive and drug-resistant cells, which neglects the continuous nature of phenotypic plasticity. In this study, we propose a mathematical model that integrates a continuous drug susceptibility index with a probabilistic inheritance function to describe clonal dynamics under therapy. The resulting integro-differential system generalizes traditional two-type competition models and captures both heterogeneity and plasticity of tumor cells. Analytical and numerical studies show that (i) continuous therapy drives rapid expansion of resistant clones, (ii) adaptive therapy maintains long-term tumor control by dynamically regulating sensitive populations, and (iii) high phenotypic plasticity accelerates phenotype switching, leading to earlier tumor relapse following continuous therapy. These results identify critical parameter regimes where adaptive therapy outperforms fixed regimens and highlight the essential role of plasticity in shaping treatment outcomes. The proposed framework provides a more realistic mathematical foundation for the design of clinically relevant adaptive therapy strategies.
\end{abstract}

\textbf{MSCcodes:\ \ } 92C50, 34K99



\section{Introduction}
\label{sec1}

Cancer remains a leading cause of death worldwide \cite{Sung_2021}. Standard treatment typically follows the maximum tolerated dose (MTD) regimen, which uses continuous high-intensity therapy to shrink tumors. Although effective in the short term, MTD often accelerates resistance by promoting the outgrowth of resistant subclones \cite{Greaves_2012}. Adaptive therapy (AT) offers an alternative treatment strategy by leveraging competition between drug-sensitive and resistant cells. It adjusts drug doses based on tumor response, aiming to sustain a sensitive cell population that suppresses resistant ones \cite{Gatenby_2009}. Nevertheless, determining how to design and optimize adaptive therapy schedules remains a central challenge as treatment outcomes are strongly influenced by tumor heterogeneity, phenotypic plasticity, and other evolutionary dynamics. 

The rationale of AT is rooted in evolutionary dynamics of tumor cell populations, where treatment is designed to exploit competitive interactions between sensitive and resistant clones \cite{Gatenby_2020, Aktipis_2013}. In contrast to MTD regimens that impose strong selective pressure and favor resistant phenotypes \cite{Gallaher_2018}, AT applies treatment intermittently or adaptively to preserve a sufficient pool of sensitive cells. These cells act as competitors that suppress the growth of resistant subpopulations, a mechanism consistent with predictions from evolutionary game theory \cite{Gillies_2012, Zhang_2017}. By harnessing this competitor-exploiter relationship, AT converts tumor heterogeneity from a clinical obstacle into a therapeutic advantage \cite{Kim_2021,Liu_2022}. Early clinical studies and pilot trials have provided encouraging evidence that AT can prolong time to progression and reduce cumulative drug exposure compared with standard MTD approaches \cite{Gatenby_2009,Zhang:2022aa}.

Mathematical modeling has become a powerful tool for analyzing and optimizing AT. By capturing tumor-drug-evolutionary dynamics, the models provide a quantitative framework to compare treatment schedules and identify conditions that prolong tumor control. Early approaches focused on population-level competition models, showing that adaptive or intermittent dosing can delay resistance relative to continuous treatment \cite{Zhang_2017, Kim_2021}. More recently, optimization-based methods have been applied to determine dosing thresholds and adaptive schedules that maximize progression-free survival while minimizing cumulative drug exposure \cite{Liu_2022, Wang_2025}. Other studies have incorporated spatial structure and pharmacokinetic-pharmacodynamic interactions, revealing how tumor heterogeneity and microenvironmental constraints shape treatment outcomes \cite{Gallaher_2018, McGehee_2024}. 

Beyond these mathematical models, several extensions have been proposed to capture additional biological mechanisms. Kam et al. \cite{Kam_2014} introduced a systematic strategy based on evolutionary dynamics, demonstrating that dynamically adjusted regimens can sustain tumor control in metastatic breast cancer. Gevertz et al. \cite{Gevertz_2025} developed a two-population model and showed that moderate-dose, beat-paced chemotherapy prolongs treatment efficacy. Wang et al. \cite{Wang_2024} integrated heterogeneity, pharmacodynamics, and immune interactions into a prostate cancer model, illustrating the combined role of intermittent dosing and immunomodulation. In parallel, West et al. \cite{West_2020} proposed a multidrug evolutionary game-theoretic framework, highlighting tumor type-dependent outcomes and the limitations of AT in aggressive malignancies.

Together, these studies underscore the critical role of mathematical models in guiding the design of adaptive therapy. However, most existing frameworks simplify tumor populations into two discrete categories---drug-sensitive versus drug-resistant cells. This binary assumption overlooks phenotypic plasticity, a continuous and dynamic process through which cells can switch between states and generate heterogeneous responses to therapy. Capturing this dimension of tumor evolution is essential for accurately predicting treatment outcomes and optimizing therapeutic strategies, and forms the focus of the present study.

To address this gap, we develop a novel framework that introduces a continuous drug susceptibility index to characterize heterogeneous treatment responses. By incorporating phenotypic plasticity through a probabilistic inheritance mechanism, the model captures dynamic phenotype switching across cell generations, providing a more realistic description of tumor evolution. Within this framework, we compare continuous therapy (CT) and adaptive therapy (AT), and systematically evaluate how plasticity affects therapy outcomes. Our simulations reveal that while CT may result in tumor recurrence due to the expansion of drug-resistant cells, AT can maintain long-term tumor control by dynamically balancing sensitive and resistant populations. Moreover, plasticity is shown to play a decisive role: high plasticity accelerates tumor relapse under CT, alters dosing requirements under AT, and can compromise treatment duration when cycle timing is suboptimal. This study provides a unified integro-differential framework that integrates tumor heterogeneity, phenotypic plasticity, and adaptive treatment response---offering a more realistic mathematical foundation for designing clinically relevant adaptive therapy strategies.

The rest of this paper is structured as follows. Section \ref{sec2} presents the formulation of the kinetic model with phenotypic plasticity. Section \ref{sec3} reports theoretical and simulation analyses of tumor dynamics under different treatment strategies, including cohort-based studies of a virtual patient cohort. Section \ref{sec4} concludes with a summary of findings, limitations, and perspectives for future research.

\section{Mathematical formulation}
\label{sec2}

\subsection{Homogeneous G0 cell cycle model}

Cancer is a disease characterized by the abnormal proliferation of cells. A classical framework for modeling proliferative dynamics is the G0 cell cycle model, originally proposed in the 1970s \cite{Burns_1970}. In this model, the cell cycle is divided into two major states: a resting phase (G0) and a proliferating phase including G1, S, G2, and mitosis (Fig. \ref{fig:1}a). During the proliferating phase, a cell may undergo apoptosis with a constant rate $\mu$, or successfully complete mitosis and produce two daughter cells after a duration $\tau$. In the resting phase, cells may either irreversibly differentiate into downstream lineages at a rate $\kappa$ or re-enter the proliferating pool at a rate $\beta$. Importantly, $\beta$ is often regulated by the abundance of resting-phase cells through cytokine-mediated feedback, reflecting a balance between self-renewal and differentiation.

\begin{figure}[htbp]
\centering
\includegraphics[width=12cm]{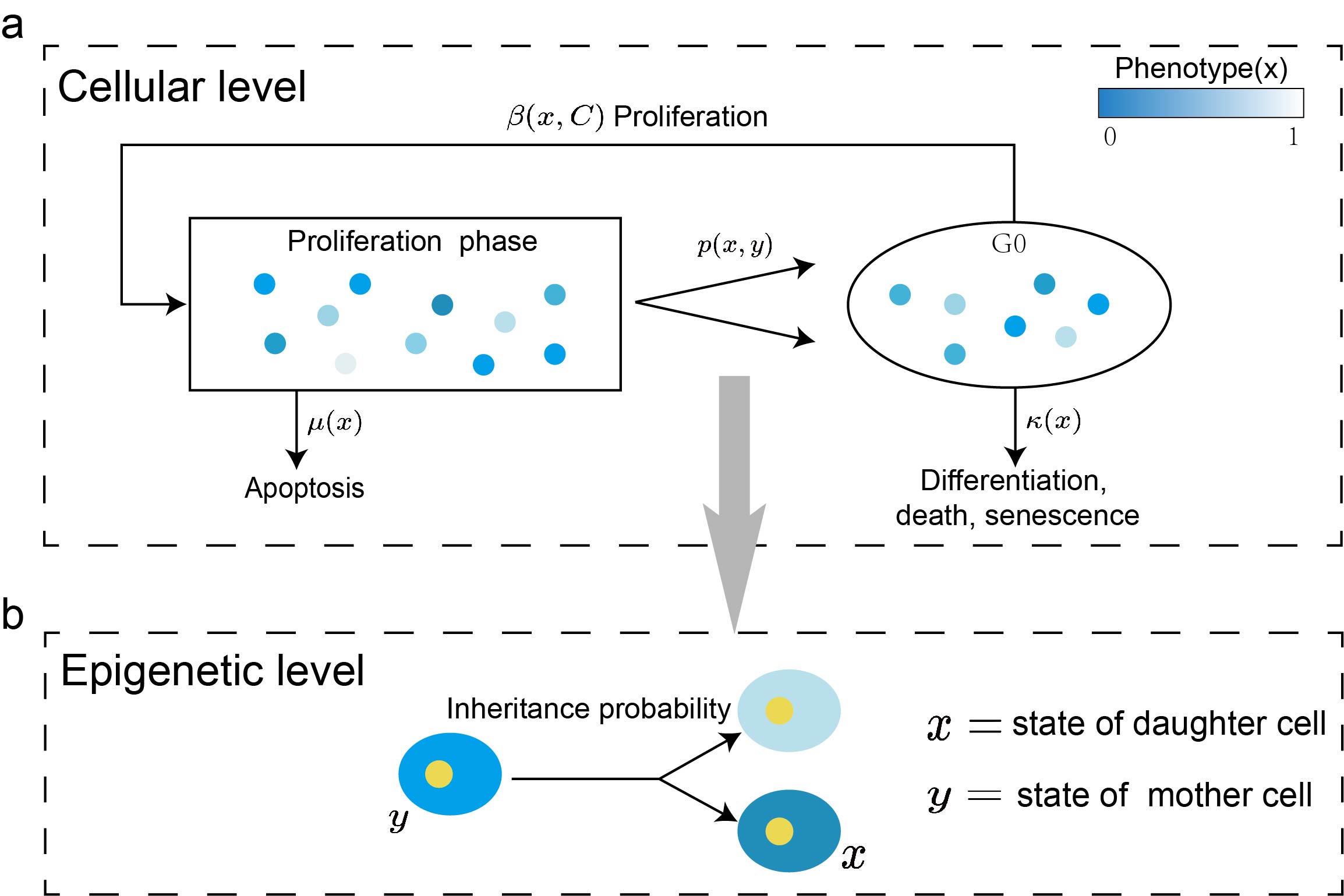}
\caption{Mechanistic diagram. (a) The G0 cell cycle model for the balance between resting and proliferating states. (b) At the epigenetic level, the probability of transmitting drug sensitivity from mother to daughter cells is described by the inheritance function $p(x, y)$.}
\label{fig:1}
\end{figure}

The above assumptions lead to the following age-structured equation \cite{Lei_2020}:
\begin{equation}
\label{eq:ag}
\left\{
\begin{aligned}
&\dfrac{\partial s(t, a)}{\partial t} + \dfrac{\partial s(t, a)}{\partial a} = -\mu s(t, a),\quad t > 0, 0<a < \tau,\\
&\dfrac{d c}{d t} = 2 s(t, \tau) - (\beta(c) + \kappa)c, t > 0,
\end{aligned}
\right.
\end{equation}
where $s(t, a)$ is the number of cells at time $t$ with age $a$ in the proliferating phase, and $c(t)$ is the number of cells in the resting phase. The boundary condition at age $a = 0$ is given by
$$
s(t, 0) = \beta(c(t))c(t).
$$

The proliferation rate $\beta$ depends on the cell number through a Hill function
\begin{equation}
\label{eq:2}
\beta(c) = \beta_0 \frac{1}{1 + (c/\theta)^m}.
\end{equation}
The Hill function represents the overall effect of feedback regulation from resting-phase cells. Biologically, the self-renewal capacity of a cell is regulated by microenvironmental factors (e.g., growth factors and cytokines) and intracellular signaling pathways. The Hill functional dependence can be derived from simple assumptions on the receptor-ligand binding dynamics \cite{Bernard_2003,Lei_2020}.  

Integrating the first equation in \eqref{eq:ag} using the method of characteristics yields a delay differential equation (DDE) \cite{Burns_1970,Mackey:1978}:
\begin{equation}
\label{eq:1}
\dfrac{d c}{d t} = -(\beta(c(t)) + \kappa) c(t) + 2 e^{-\mu \tau} \beta(c(t-\tau)) c(t-\tau).
\end{equation}
Equation \eqref{eq:1} describes the dynamics of homogeneous tumor cell growth. 

\subsection{Heterogeneous tumor cell sensitivity}

To incorporate treatment heterogeneity, we introduce a continuous variable $x\in [0, 1]$ representing drug sensitivity, with larger $x$ values corresponding to higher sensitivity. Let $C(t, x)$ denote the density of cells with sensitivity $x$ at time $t$. The total tumor population is then 
$$
c(t) = \int_0^1 C(t, x) d x.
$$

Drug-sensitive cells typically proliferate faster, so tumors are dominated by sensitive subclones before therapy. Under drug pressure, however, resistant cells gain a survival advantage. Hence, both the proliferation rate $\beta$ and death rate $\mu$ are dependent on the sensitivity index $x$, and the cycle duration $\tau(x)$ may also vary with sensitivity.

Biologically, sensitivity reflects the overall impact of drug sensitivity based on a cell's molecular profile, which includes the expression of drug targets, efflux pumps, and stress response pathways. These features may fluctuate during cell division, giving rise to phenotypic plasticity. To phenomenologically describe the phenotypic plasticity, we introduce an inheritance probability function $p(x, y)$, representing the conditional probability that a daughter cell has sensitivity $x$ given that its mother cell had sensitivity $y$ (Fig. \ref{fig:1}b):
$$
p(x, y) = P(\mbox{daughter cell} = x \vert \mbox{mother cell} = y).
$$

This leads to a modified age-structured equation:
\begin{equation}
\label{eq:ag2}
\left\{
\begin{aligned}
&\dfrac{\partial s(t, a, x)}{\partial t} + \dfrac{\partial s(t, a, x)}{\partial a} = -\mu(x) s(t, a, x), \quad t > 0,\ 0< a < \tau(x),\\
&\dfrac{\partial C(t, x)}{\partial t} = 2 \int_0^1 s(t, \tau(y), y) d y - (\beta(x, C(t, x)) + \kappa(x)) C(t, x),\quad t > 0.
\end{aligned}
\right.
\end{equation}
Applying the method of characteristics, we obtain an integro-differential equation \cite{Lei_2020, Lei_2020gd}:
\begin{equation}
\label{eq:4}
\begin{aligned}
\frac{\partial C(t, x)}{\partial t} &= -(\beta(x,C(t, x))+\kappa(x))C(t, x)\\
&{}+2\int_{0}^{1}e^{-\mu(y)\tau(y)}\beta(y, C(t-\tau(y), y))C(t-\tau(y), y)p(x,y)dy.
\end{aligned}
\end{equation}
Equation \eqref{eq:4} integrates stem cell regeneration and the transition of sensitivities during cell division, which form the basis of our analysis of adaptive therapy with tumor cell plasticity.

For comparison with classical ODE models, we simplify \eqref{eq:4} by omitting the explicit delay $\tau$ and absorbing proliferative-phase death into $\mu(x)$, giving
\begin{equation}
\label{eq:5}
\frac{\partial C(t, x)}{\partial t} = -(\beta(x,C(t, x))+\kappa)C(t, x)+2\int_{0}^{1}e^{-\mu(y)}\beta(y, C(t, y))C(t, y)p(x,y)dy.
\end{equation} 

\subsection{Subpopulational competition}
To distinguish between drug-sensitive and drug-resistant subgroups based on the heterogeneity in drug sensitivity, we define a logistic function $\xi(x)$ that specifies the probability of a cell with sensitivity $x$ belonging to the resistant group:
\begin{equation}
\label{eq:xi}
\xi(x) = \frac{1}{1+\exp((x-k_1)/s_1)}.
\end{equation}
Here, the parameter $k_1$ is the location parameter, representing the threshold for distinguishing between drug-sensitive phenotypes ($x > k_1$) and drug-resistant phenotypes ($x < k_1$), and $s_1$ is the scale parameter that controls the steepness of the logistic curve. Therefore,    
$$
c_0(t) = \int_0^1 C(t, x) \xi(x) d x,  \quad c_1(t) = \int_0^1 C(t, x) (1 - \xi(x)) d x
$$
represent resistant and sensitive populations, respectively. 

The two groups of cells compete for resources. Before treatment, sensitive cells dominate and suppress resistant clones. Under prolonged therapy, resistant cells expand and suppress sensitive clones. To capture this ecological interaction, we modify the proliferation rate as:
$$
\beta(x, C(t, x)) = \beta_0 \dfrac{1}{1 + \left(\frac{c_0(t)}{\theta/a_0(x)}\right)^m + \left(\frac{c_1(t)}{\theta/a_1(x)}\right)^m},
$$ 
where $a_0(x)$ and $a_1(x)$ are the competition coefficients for resistant and sensitive cells, respectively, defined by Hill functions:
\begin{equation}
\label{eq:9}
\begin{aligned}
&a_0(x) = \rho_0 + (1-\rho_0)\dfrac{x^8}{k_1^8 + x^8},\\
&a_1(x) = \rho_1 + (1-\rho_1) \dfrac{k_1^8}{k_1^8 + x^8}.
\end{aligned}
\end{equation}
Here, $k_1$, as referenced in \eqref{eq:xi}, serves as the threshold for distinguishing competition interactions between sensitive and resistant cells. The parameters $\rho_0$ and $\rho_1$ define the baseline competition coefficients. Thus, resistant cells (small $x$) are mainly repressed by sensitive competitors, whereas sensitive cells (large $x$) are increasingly suppressed as resistant clones accumulate. 

The governing system is then:
\begin{equation}
\label{eq:dio}
\left\{
\begin{aligned}
\dfrac{\partial C(t, x)}{\partial t} &= -(\beta(x, C) + \kappa) C(t, x) + 2 \int_0^1 e^{-\mu(t, y)} \beta(y, C) C(t, y) p(x, y) d y,\\
\beta(x,C) &= \beta_0 \dfrac{1}{1+(\frac{c_0(t)}{\theta/a_0(x)})^m+(\frac{c_1(t)}{\theta/a_1(x)})^m},\\
c_0(t) &= \int_0^1 C(t, x) \xi(x) d x,\\
c_1(t) &= \int_0^1 C(t, x) (1- \xi(x)) d x.
\end{aligned} 
\right.
\end{equation}
Here, $\mu(t, x)$ accounts for the time-dependent treatment protocol.

\subsection{Cell plasticity}

The genetic function $p(x,y)$ describes the stochastic fluctuations of sensitivity traits during cell division, reflecting the imperfect transmission of molecular and epigenetic states from parent to daughter cells. Although the underlying biochemical mechanisms are highly complex, employing a probabilistic description provides a manageable approximation framework \cite{Lei_2020gd,Zhang:2021gd}. 

Since $x \in [0, 1]$, we represent $p(x, y)$ by a beta-distribution density 
\begin{equation}
\begin{aligned}
p(x,y)=\dfrac{x^{a(y)-1}(1-x)^{b(y)-1}}{B(a(y),b(y))},
B(a,b)=\dfrac{\Gamma(a)\Gamma(b)}{\Gamma(a+b)},
\end{aligned}
\end{equation}
where $\Gamma(\cdot)$ is the gamma function, and $a(y), b(y)$ are shape parameters depending on mother-cell sensitivity $y$. The beta distribution is used because it is ideal for bounded domains and can effectively represent various inheritance patterns, ranging from conservative transmission to random fluctuations, by adjusting its shape parameters. Additionally, different shapes of the probability function can be obtained by varying these shape parameters \cite{Lei_2020,Huang:2024aa}.

To determine the shape parameters for a beta distribution, we note that the mean and variance are
\begin{equation}
\mathrm{E}[X] = \dfrac{a}{a + b}, \quad \mathrm{Var}[X] = \dfrac{a b}{(a + b)^2 (a + b - 1)}.
\end{equation}
We introduce two functions $\phi(y)$ and $\eta(y)$ to express the conditional expectation and variance as
\begin{equation}
\label{eq:11}
\mathrm{E}[x \vert y] =\phi(y),\quad \mathrm{Var}[x \vert y]=\dfrac{1}{1+\eta(y)}\phi(y)(1-\phi(y)).
\end{equation}
This gives
$$
a(y) = \eta(y) \phi(y),\quad b(y) = \eta(y) (1 - \phi(y)).
$$
Here, $\phi(y)$ represents the expected value of offspring sensitivity given parental sensitivity $y$, which reflects the maintenance of phenotypic stability through epigenetics. Meanwhile, $\eta(y)$ regulates the dispersion of this mean, capturing the randomness induced by gene expression noise and microenvironmental fluctuations. 

For simplicity, $\eta(y)$ is often taken as a constant, while $\phi(y)$ is typically given by an increasing Hill function \cite{Lei_2020gd,Li:2025kf}. Here, we define $\phi(y)$ as
\begin{equation}
\phi(y) = \phi_0+\phi_1\frac{y^2}{k_2^2+y^2}, 
\end{equation}
where $k_2$ is the half-saturation coefficient. The parameters are taken so that $\eta > 0$ and $0< \phi(y) < 1$.

Biologically, the functions $\phi(y)$ and $\eta(y)$ have clear biological interpretations and can, in principle, be determined through experimental observations by tracking cell cycling and single-cell sequencing \cite{Quinn:2021aa,Denoth-Lippuner:2021aa}. We can also derive the function $\phi(y)$ based on in silico dynamics of cell regeneration \cite{Li:2025kf}. In this study, the parameter values within these two functions were selected based on the following criteria: (i) consistency with biologically plausible ranges reported in the literature for similar phenotype conversion models \cite{Lei_2020gd, Zhang:2021gd, Liang:2024kg, Wang:2025hl}; and (ii) ensuring that the resulting dynamics can replicate clinically observed behavior, specifically tumors initially exhibiting sensitivity to treatment but eventually developing resistance under sustained therapy. 

By constructing the probability kernel function $p(x,y)$ using a beta distribution, this framework translates molecular-level variation into population-level phenotypic plasticity. It simultaneously characterizes stable phenotypic transmission and random transitions between sensitive and resistant states. This model framework has been applied in studies of tumor cell immune escape \cite{Zhang:2021gd} and cancer drug resistance \cite{Wang:2025hl}.

It is important to note that in this study, the inheritance function $p(x,y)$ is assumed to remain unchanged throughout treatment. This assumption simplifies the complexities of the real situation by disregarding the effect of treatment on phenotypic plasticity. Alternatively, to account for such effects, the functions $\phi(y)$ could be made dependent on drug administration; however, exploring this further is beyond the scope of this study.

In summary, our formulation allows us to tune the balance between accurate inheritance (low variance) and stochastic switching (high variance), thereby connecting molecular noise during division to macroscopic evolution.

\subsection{Treatment protocol} 
Therapeutic interventions are modeled as a time-dependent increase in cell death rate, expressed as:
$$
\mu(t, x) = \mu_0 + \mu_D(t, x),
$$
where $\mu_0$ represents the baseline death rate, and $\mu_D(t, x)$ denotes the drug-induced component that depends on the dosage and sensitivity of the cells. Since highly sensitive cells are more affected by the treatment, we describe this dose-response relationship using a Hill function:
\begin{equation}
\label{eq:7}
\mu_D(x)=\mu_1(t)\dfrac{x^2}{x^2+k_1^2}.
\end{equation}
In this equation, $k_1$, as referenced in \eqref{eq:xi} and \eqref{eq:9}, serves as the threshold to differentiate the drug effects on sensitive versus resistant cells. The term $\mu_1(t)$ represents the dosing schedule. This Hill-type function captures the pharmacodynamic relationship between drug concentration and killing efficacy, ensuring that sensitive cells, which are characterized by higher values of $x$, experience greater cell death than resistant cells.

The heterogeneous framework with phenotypic plasticity serves as the foundation of our analysis. In the next subsection, we show that under specific assumptions, the model reduces to classical formulations, which will later be used as benchmarks for comparison. 

\subsection{ODE model}
\label{sec2.3}

To provide a benchmark for comparison, we simplify the heterogeneous model that incorporates phenotypic plasticity into a more straightforward ordinary differential equation (ODE) system. This reduced formulation aligns with the classical two-population framework commonly used in previous studies of adaptive therapy. We eliminate plasticity by setting the inheritance function $p(x, y) = \delta(x-y)$, meaning that daughter cells inherit the same drug sensitivity as their mothers without variation. We then group the tumor population into two distinct subgroups: drug-resistant $c_0(t)$ and drug-sensitive cells $c_1(t)$. For simplicity, we also assume that the differentiation rate $\kappa(x)$ is a constant, expressed as $\kappa(x) = \kappa$, regardless of sensitivity. 

To capture subgroup-specific competition, the competition coefficients $a_0(x)$ and $a_1(x)$ are defined as piecewise constants:
\begin{equation}
\label{eq:13}
a_0(x)=
\begin{cases}
1, & x>k_1\\
\rho_0, & x<k_1
\end{cases} ,\quad
a_1(x)=
\begin{cases}
\rho_1, & x>k_1\\
1, & x<k_1
\end{cases}.
\end{equation}
Biologically, this means that resistant cells compete less effectively when surrounded by sensitive cells (scaled by $\rho_0$), while sensitive cells become increasingly suppressed as resistant subclones expand (scaled by $\rho_1$).

The therapeutic effect is modeled by a piecewise dose-response killing function:
\begin{equation}
\label{eq:12}
\mu_D(t, x)=
\begin{cases}
\mu_1(t), &  x>k_1\\
 0 , &  x<k_1
\end{cases},
\end{equation}
which assumes that only the sensitive population is effectively targeted by treatment, whereas resistant cells remain unaffected. In practice, this reflects the clinical scenario in which resistant clones evade drug-induced apoptosis through mechanisms such as efflux pumps, mutations in drug targets, or altered signaling pathways.

 Under these assumptions, substituting the piecewise competition terms and dose-response function into the general integro-differential system \eqref{eq:dio}, we obtain the following ODE system:
\begin{equation}
\label{eq:17}
\left\{
\begin{aligned}
\dfrac{\mathrm{d} c_0}{\mathrm{d} t} &= \left(\beta_0 \dfrac{(2e^{-\mu_0}-1)}{1 + (\dfrac{c_0}{\theta/\rho_0})^m + (\dfrac{c_1}{\theta})^m}-\kappa\right)c_0 ,\\
\dfrac{\mathrm{d} c_1}{\mathrm{d} t} &= \left(\beta_0 \dfrac{(2e^{-(\mu_0+\mu_1(t))}-1)}{1 + (\dfrac{c_0}{\theta})^m + (\dfrac{c_1}{\theta/\rho_1})^m}-\kappa\right)c_1.
\end{aligned}
\right.
\end{equation}

This system illustrates ecological competition between resistant and sensitive subclones, in which each population suppresses the other's growth. Before treatment, sensitive cells usually dominate due to their higher proliferation rates. However, with prolonged therapy, resistant clones gradually emerge and suppress the sensitive population, ultimately leading to treatment failure. This model aligns with classical competition frameworks proposed by Gatenby and colleagues \cite{Gatenby:1991oh, Gatenby_2009} and has been extensively used to investigate adaptive therapy. In the following analysis, we compare this ODE model with our heterogeneous plasticity framework to highlight the impact of phenotypic switching on treatment outcomes.

\section{Results}
\label{sec3}
\subsection{Kinetic analysis and competition dynamics}
\label{sec3.1}

\subsubsection*{Mathematical analysis}

We begin by analyzing the classical ODE model \eqref{eq:17} to understand the dynamics of tumor cells under continuous treatment. To focus on equilibrium behavior, we assume a constant drug effect by setting $\mu_1(t) = \mu_1$. Furthermore, we introduce normalized state variables $x_0$ and $x_1$ to represent the scaled densities of resistant and sensitive populations, respectively:
$$x_0(t)=\left(\frac{c_0(t)}{\theta}\right)^m, \quad x_1(t)=\left(\frac{c_1(t)}{\theta}\right)^m.$$ 
For convenience, we define the following parameters: 
$$\alpha_0=\frac{\beta_0 (2e^{-\mu_0} - 1)}{\kappa},\qquad \alpha_1 = \frac{\beta_0 (2e^{-(\mu_0 + \mu_1)} - 1)}{\kappa},$$
where $\alpha_0$ and $\alpha_1$ characterize the effective proliferative potential of resistant and sensitive cells, respectively, relative to differentiation loss $\kappa$. With these substitutions, the system \eqref{eq:17} can be rewritten as
\begin{equation}
\label{eq:18}
\left\{
\begin{aligned}
\dot{x}_0&= m \kappa \left(\frac{\alpha_0}{1 + \rho_0^{m}x_0 + x_1}-1\right) x_0,\\
\dot{x}_1&= m \kappa \left(\frac{\alpha_1}{1 + x_0 + \rho_1^{m}x_1}-1\right) x_1.
\end{aligned}
\right.
\end{equation}

For biological relevance, we impose the following parameter restrictions:
\begin{equation}
\label{eq:cond}
\alpha_0 > 0,\ \alpha_0 > \alpha_1,\  1> \rho_0,\ \rho_1 > 0.
\end{equation}
That is, resistant cells ($c_0$) retain a higher proliferation advantage under therapy, but both populations are subject to competition-mediated inhibition.

The system admits several steady states, representing different ecological outcomes of tumor evolution:
$$
\begin{aligned}
 E_0 &= (0,0),\ E_1 = (\dfrac{\alpha_0-1}{\rho_0^m}, 0),\ E_2 = (0 , \dfrac{\alpha_1-1}{\rho_1^m}),\\
 E_3 &= \left(\dfrac{(\alpha_1-1)-\rho_1^m(\alpha_0-1)}{1-(\rho_0 \rho_1)^m} , \dfrac{(\alpha_0-1)-\rho_0^m(\alpha_1-1)}{1-(\rho_0 \rho_1)^m}\right).
\end{aligned}
$$

\begin{theorem}
\label{thm:1}
Consider equation \eqref{eq:18} under the restrictions \eqref{eq:cond}. The following holds:
\begin{enumerate}
\item[(1)] When $1 > \alpha_0 \geq \alpha_1$, the only non-negative steady state is $E_0$, which is locally asymptotically stable.
\item[(2)] When $\alpha_0 > 1 > \alpha_1$, there exist two steady states $E_0$ and $E_1$. Among them, $E_0$ is unstable, while $E_1$ (resistant-dominant) is locally asymptotically stable.
\item[(3)] When $\alpha_0 \geq \alpha_1 > 1$, three steady states $E_0$, $E_1$, and $E_2$ exist. If additionally, 
\begin{equation}
\label{eq:cond2}
\rho_1^m < \frac{\alpha_1  - 1}{\alpha_0 - 1},
\end{equation}
a coexistence state $E_3$ also appears. In these cases, $E_0$ and $E_3$ are unstable, $E_1$ is locally asymptotically stable, and $E_2$ is stable only when \eqref{eq:cond2} is satisfied.
\end{enumerate}
\end{theorem}
\textbf{Proof}\ \ 
The steady states of \eqref{eq:18} are given by the solutions of the following equation
\begin{equation}
\label{eq:19}
\left\{
\begin{aligned}
 m \kappa (\frac{\alpha_0}{1+\rho_0^{m}x_0+x_1}-1) x_0 = 0 \\
 m \kappa (\frac{\alpha_1}{1+x_0+\rho_1^{m}x_1}-1) x_1 = 0 
\end{aligned}
\right..
\end{equation}
It is simple to solve the above equations to find the solutions $E_0, E_1, E_2, E_3$.

It is straightforward that $E_1$ (or $E_2$) is a non-negative steady state only when $\alpha_0 >1$ (or $\alpha_1 > 1$). Moreover, $E_3$ is a non-negative steady state if and only if $\alpha_0 > \alpha_1 > 1$ and \eqref{eq:cond2} is satisfied. Consequently, the existence of non-negative steady states under various conditions is readily established.

To further study the stability of the steady states, the Jacobian matrix of \eqref{eq:18} at the state $E_i$ is given by 
\begin{equation}
\label{eq:23}
J(E_i) = m \kappa \left.\begin{pmatrix}
(\frac{\alpha_0(1 + x_1)}{(1 + \rho_0^m x_0 + x_1)^2}-1) & -\frac{\alpha_0 x_0}{(1 + \rho_0^m x_0 + x_1)^2}  \\
- \frac{\alpha_1 x_1}{(1 + x_0 + \rho_1^m x_1)^2} & (\frac{\alpha_1(1+x_0)}{(1 + x_0 + \rho_1^m x_1)^2}-1)
\end{pmatrix}\right\vert_{E_i}.
\end{equation}

For the state $E_0$, we have
$$
J(E_0) = m \kappa \begin{pmatrix}
(\alpha_0-1) & 0 \\
0 & (\alpha_1-1)
\end{pmatrix},
$$
which yields the eigenvalues $\alpha_0 -1$ and $\alpha_1 - 1$. Thus, $E_0$ is locally asymptotically stable only when $\alpha_0,\ \alpha_1 < 1$.

For the state $E_1$, we always have $\alpha_0 > \alpha_1 > 1$, and 
$$
J(E_1) =  m \kappa \left.\begin{pmatrix}
\frac{1}{\alpha_0} - 1 & - \frac{x_0}{\alpha_0}\\
0 & \frac{\alpha_1}{1 + x_0} - 1
\end{pmatrix}\right\vert_{x_0 = \frac{\alpha_0 - 1}{\rho_0^m}},
$$
which yields eigenvalues $\lambda_1 = \frac{1}{\alpha_0} -1$ and $\lambda_2 = \frac{\alpha_1}{1 + x_0} - 1$. Since we always have $\alpha_0 - 1 > \rho_0^m (\alpha_1 - 1)$, $E_1$ is local asymptotically stable only when $\alpha_0 > 1$.

For the state $E_2$, we have
$$
J(E_2) = m \kappa \left.\begin{pmatrix}
 \frac{\alpha_0}{1 + x_1} - 1 & 0 \\
 - \frac{x_1}{\alpha_1} & \frac{1}{\alpha_1} - 1
\end{pmatrix}\right\vert_{x_1 = \frac{\alpha_1 - 1}{\rho_1^m}},
$$
which yields eigenvalues $\lambda_1 = \frac{\alpha_0}{1 + x_1} -1$ and $\lambda_2 = \frac{1}{\alpha_1} - 1$. Thus, $E_2$ is locally asymptotically stable only when $\alpha_1 > 1$ and \eqref{eq:cond2} is satisfied.

For the state $E_3$, we have
$$
J(E_3) = m \kappa \left.\begin{pmatrix}
\frac{\rho_0^m x_0}{\alpha_0} & - \frac{x_0}{\alpha_0}\\
-\frac{x_1}{\alpha_1}  & \frac{\rho_1^m x_1}{\alpha_1}
\end{pmatrix}\right\vert_{E_3}.
$$
Since 
$$
\mathrm{tr}(J(E_3)) = \frac{\rho_0^m x_0}{\alpha_0} + \frac{\rho_1^m x_1}{\alpha_1} > 0,
$$
the state $E_3$ is always unstable. 

Summarizing the above arguments, the stabilities of the steady states are concluded.

Theorem \ref{thm:1} demonstrates how equilibrium outcomes are influenced by the balance among proliferation, drug-induced death, and competition. Biologically, state $E_2$ represents tumors dominated by sensitive cells in the absence of treatment, which aligns with the initial high responsiveness to therapy. However, with prolonged exposure to the drug, the system shifts toward state $E_1$, where resistant cells become dominant. The unstable coexistence state $E_3$ indicates that even when both subclones are present, competitive imbalances generally lead to one population prevailing over the other. 

In our simulations, we aim to mimic the clinically observed situation in which tumors are initially drug-sensitive but develop resistance during ongoing therapy. To achieve this, we have selected parameters such that $E_2$ is stable before treatment (with $\mu_1 = 0$), whereas $E_1$ becomes stable once therapy is initiated (with $\mu_1 > 0$). The default parameter values used in our study are listed in Table \ref{table:1}. 

\begin{table}[htbp]
\centering
\begin{tabular}{c c c l}
\hline
Parameter & Value & Unit & Description \\  
\hline
$\beta_0$ & $0.046$ & $\mathrm{h}^{-1}$ & Maximum proliferation rate of tumors\\
$\theta$ & $10^4$ & cells & Inhibitory threshold between cell populations\\
$\kappa$ & $0.005$ & $\mathrm{h}^{-1}$ & Differentiation rate of tumor cells \\
 $\mu_0$ & $0.08$ & $\mathrm{h}^{-1}$ & Baseline death rate of tumor cells  \\
 $k_1$ & $0.3$ & - & Threshold constant between drug-resistant and sensitive cells \\
 $s_1$ & $0.035$ & - & Scale parameter of the logistic curve\\ 
 $\rho_0$ & $0.2$ & - & Basal competition coefficient for resistant cells\\
 $\rho_1$ & $0.2$ & - & Basal competition coefficient for sensitive cells\\
 $\phi_0$ & $0.06$ & - & Coefficient to define the function $\phi(y)$\\
 $\phi_1$ & $0.9$ & - & Coefficient to define the function $\phi(y)$\\
 $k_2$ & $0.5$ & - & Coefficient to define the function $\phi(y)$\\
$\eta$ & $60$ & - & Variance regulator of inheritance distribution\\
\hline
\end{tabular}
\caption{Default parameter values.}
\label{table:1}
\end{table}

In the simulations below, we vary the cell death rate after treatment, denoted as $\mu_1$, to investigate how tumors respond to different drug dosages. With these values, we determine $\alpha_0 = 7.79$, which ensures that $E_1$ remains stable. Additionally, when $\mu_1$ exceeds $0.51$, we find that $\alpha_1 > 1$, leading to the violation of condition \eqref{eq:cond2}. As a result, $E_2$ becomes unstable under continuous therapy.

\subsubsection*{Treatment response of competing models} 
We next investigate the therapeutic response of the classical competition model, which includes both drug-sensitive and drug-resistant cell populations. This framework assumes that resistance is a fixed trait and that tumor evolution is primarily driven by selection and competitive release under treatment pressure. To understand how therapeutic intensity affects tumor control, we first examined tumor dynamics under continuous therapy (CT) across a range of the drug efficacy parameter $\mu_1$.

Without treatment, the sensitive cell population expanded exponentially, dominating the tumor mass. In contrast, the resistant cells remained at low levels and eventually became extinct due to their disadvantage in proliferation (Fig. \ref{fig:2}a). When continuous therapy began after the total tumor burden reached $10^5$ cells, the system displayed dose-dependent evolutionary dynamics. With moderate drug efficacy ($\mu_1 = 0.4$), the selective pressure was not strong enough to provide a significant advantage to the resistant subpopulation, and the tumor continued to be dominated by sensitive clones (Fig. \ref{fig:2}b). However, when drug efficacy increased to $\mu_1 = 0.6$, the sensitive population began to decline sharply in a dose-dependent manner, while resistant clones expanded due to competitive release. The reduced proliferative suppression from sensitive cells allowed the resistant population to grow (Fig. \ref{fig:2}c). Consequently, the total tumor burden initially decreased but then rebounded, ultimately leading to therapy-induced relapse.

\begin{figure}[htbp]
\centering
\includegraphics[width=12cm]{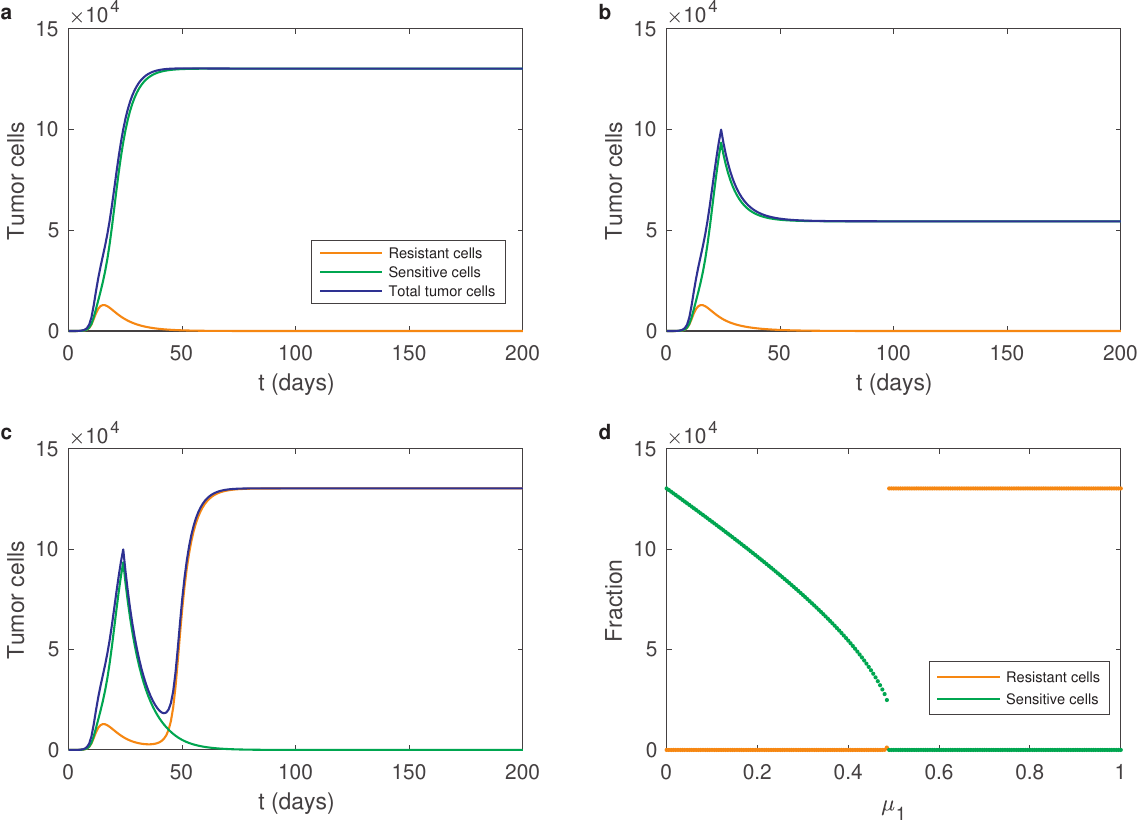}
\caption{Evolution of tumor cell dynamics in the classical competition model. (a) Changes in cell number without therapy. (b) Changes in cell number under continuous therapy with $\mu_1 = 0.4$ (c) Changes in cell number under continuous therapy with $\mu_1 = 0.6$. (d) Steady-state populations of drug-sensitive and resistant cells after continuous therapy, with $\mu_1$ varying from $0$ to $1$. The initial conditions for all simulations are $(c_0, c_1) = (1, 1.1)$.}
\label{fig:2}
\end{figure}

These results highlight a crucial paradox in continuous therapy: increasing the drug strength does not always lead to better tumor control. A systematic dose-response analysis (Fig. \ref{fig:2}d) showed that as $\mu_1$ increased from $0$ to $1$, the steady-state population of sensitive cells decreased steadily, while resistant cells remained suppressed only when $\mu_1$ was less than $0.5$. Near the critical threshold of $\mu_1 \approx 0.50$, resistant clones began to escape suppression, resulting in a relapse. Beyond this point, further increases in $\mu_1$ did not reduce the total tumor burden---the tumor size after treatment could exceed its pre-treatment level due to the dominance of highly resistant cells. This counterintuitive behavior aligns with experimental and theoretical studies on evolutionary therapy, demonstrating that overly aggressive treatment can accelerate the evolution of resistance by eliminating competitive constraints. 

To overcome this limitation, we next implemented adaptive therapy (AT)---a treatment strategy that modulates drug administration based on tumor burden feedback. The principle of AT is to maintain a dynamic equilibrium between sensitive and resistant populations rather than attempting complete eradication. Treatment begins when the total tumor cell number reaches an initial burden $c_0 = 10^5$, and is paused once it decreases to $\alpha c_0$ ($\alpha < 1$). When the tumor regrows to $c_0$, therapy resumes. This on-off control rule is mathematically represented as:
\begin{equation}
\mu_1=
\begin{cases}
0.6, & \text{when}\ c(t) = c_0,\\
0, & \text{when}\ c(t)=\alpha c_0,
\label{eq:26}
\end{cases}
\end{equation}
where the parameter $\alpha$ governs the threshold of drug withdrawal.

\begin{figure}[htbp]
\centering
\includegraphics[width=12cm]{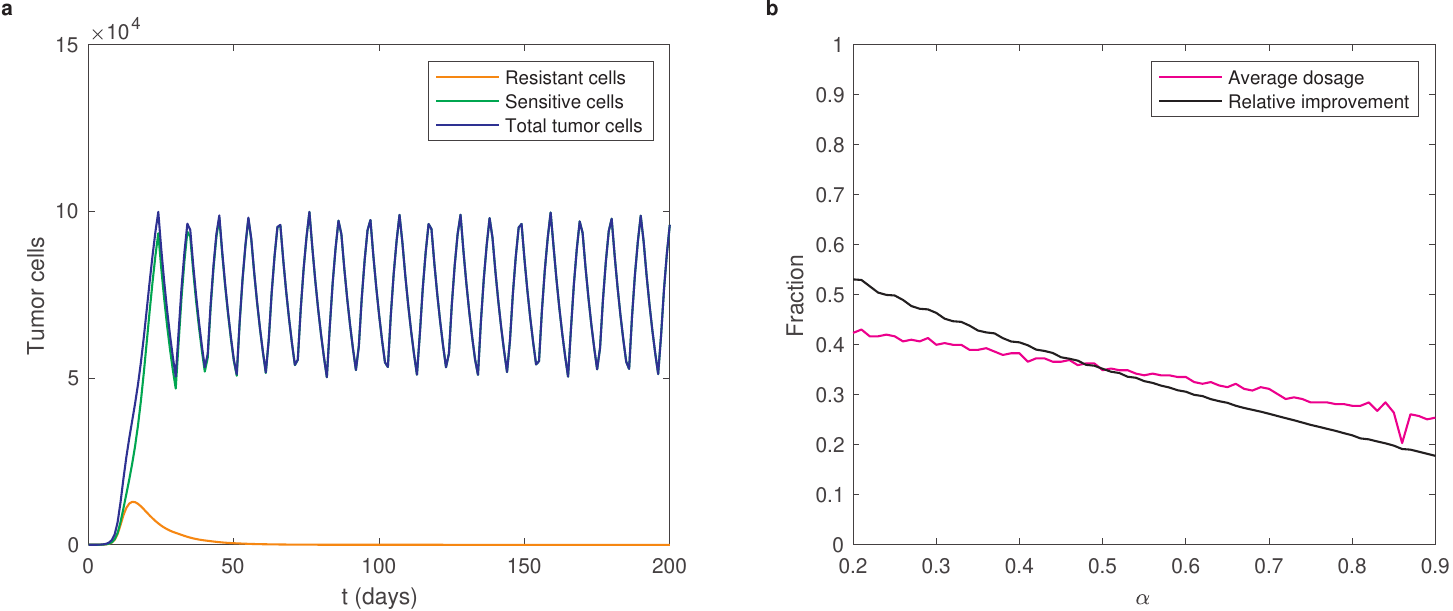}
\caption{Tumor dynamics under adaptive therapy in the competition model. (a) Temporal evolution of total tumor cells under adaptive therapy with $c_0 = 10^5$, $\alpha = 0.5$, and $\mu_1 = 0.6$. (b) Dependence of average dosage and the relative improvement of adaptive therapy on the control parameter $\alpha$.}
\label{fig:3}
\end{figure}
Under adaptive therapy, the tumor burden oscillates between $\alpha c_0$ and $c_0$, maintaining a stable cycle of regression and regrowth (Fig. \ref{fig:3}a). During the ``off'' phase, sensitive cells recover and competitively suppress resistant clones, thereby preventing their fixation. This dynamic equilibrium effectively prolongs the time to resistance while using a lower cumulative drug exposure. 

To quantitatively evaluate the efficacy of AT, we introduced two indices:
\begin{enumerate}
\item \textbf{Average dosage}, which measures the long-term mean drug application intensity:
\begin{equation}
\mbox{Average dosage} =  \lim_{T\to\infty}\dfrac{1}{T} \int_{t_0}^{t_0+T} \mu_1(t)\mathrm{d} t;
\end{equation}
and 
\item \textbf{Relative improvement}, which compares tumor control efficiency between AT and CT over an extended period:
\begin{equation}
\label{eq:ei}
\mbox{Relative improvement} = 1- \lim_{T\to\infty}\dfrac{\int_{t_0}^{t_0 + T} c_{AT}(t) \mathrm{d}t}{\int_{t_0}^{t_0 + T} c_{CT}(t) \mathrm{d} t}.
\end{equation}
\end{enumerate}
Here, $t_0$ represents the onset of treatment, and $c_{AT}(t)$ and $c_{CT}(t)$ denote the total tumor burden under adaptive and continuous therapy, respectively.

As illustrated in Fig. \ref{fig:3}b, the average dosage of the treatment decreases as the parameter $\alpha$ increases, while the relative improvement rises. Biologically, this indicates a trade-off between treatment effectiveness and drug toxicity. Lower values of $\alpha$ lead to more aggressive control cycles, which result in better tumor suppression but require a higher cumulative dosage. Therefore, it is essential to optimize $\alpha$ to find the right balance between treatment efficacy and sustainability.

The results indicate that the classical competition model effectively captures the essential evolutionary features of adaptive therapy, including dynamic control of tumor burden, maintenance of sensitive cell clones, and suppression of resistant outgrowth through competitive inhibition. However, this model assumes static phenotypes and does not account for phenotypic plasticity, a critical mechanism by which tumor cells adjust to environmental stress and temporarily acquire drug tolerance. The limitations of the fixed-trait approach highlight the need for a framework that incorporates plasticity to more accurately reflect the observed persistence of resistance and the variability in therapeutic responses in clinical settings.

\subsection{Integrating the effect of cell plasticity on tumor dynamics} 
\label{sec3.2}

To further explore how tumor cell plasticity regulates population dynamics and therapeutic response, we employed the plasticity-enabled model \eqref{eq:dio}, which characterizes the continuous evolution of drug susceptibility through an epigenetic inheritance mechanism. In this framework, each tumor cell possesses a continuously distributed drug-sensitivity index, and transitions between phenotypes occur stochastically during proliferation according to the inheritance probability function. This formulation allows for phenotypic drift, cells can reversibly transition between drug-sensitive and drug-resistant states. This phenotypic drift reflects the dynamic adaptability of tumor populations under treatment stress. 

We conducted numerical simulations under three therapeutic scenarios: without therapy, continuous therapy (CT), and adaptive therapy (AT).

As in the classical competition model, drug-sensitive cells exhibit a proliferative advantage in the absence of treatment, and drug-resistant cells are suppressed through ecological competition (Fig. \ref{fig:4}a). However, in contrast to the ODE model, resistant cells in the plasticity framework do not go extinct; instead, they persist at a low but stable fraction. This steady coexistence results from reversible phenotypic transitions between sensitive and resistant states, ensuring that a small resistant subpopulation is continuously replenished even in the absence of drug pressure.

\begin{figure}[htbp]
\centering
\includegraphics[width=12cm]{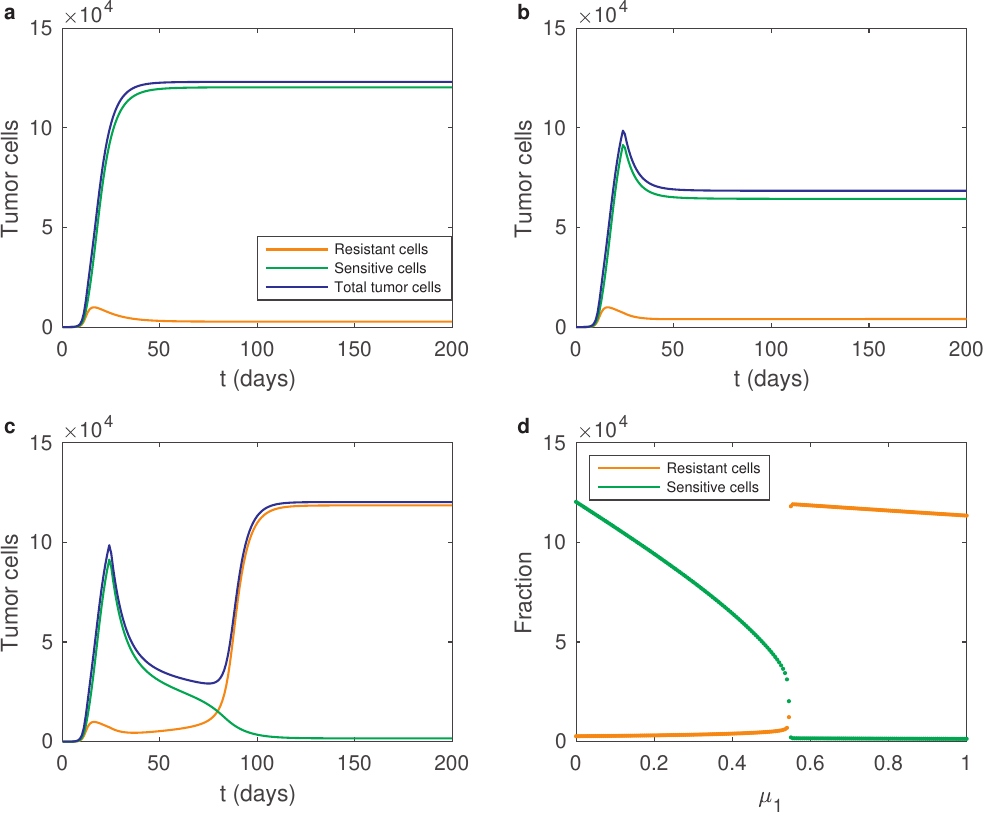}
\caption{Evolution of tumor cell dynamics in the plasticity-enabled model. (a) Changes in cell number without therapy. (b) Changes in cell number under continuous therapy with $\mu_1 = 0.4$. (c) Changes in cell number under continuous therapy with $\mu_1 = 0.6$. (d) Steady-state populations of drug-sensitive and resistant cells after continuous therapy with $\mu_1$ varying from $0$ to $1$.}
\label{fig:4}
\end{figure}

Under continuous therapy, the system exhibits dose-dependent evolutionary dynamics that are qualitatively similar to those of the ODE model. For moderate drug efficacy ($\mu_1 = 0.4$), tumor growth remains dominated by sensitive clones (Fig. \ref{fig:4}b). As the drug efficacy increases ($\mu_1 = 0.6$), sensitive cells are rapidly depleted, followed by the expansion of resistant clones through competitive release (Fig. \ref{fig:4}c). However, the regrowth phase occurs more gradually than in the non-plastic model, indicating that plasticity smoothens the evolutionary transition between sensitive and resistant phenotypes. 

A systematic dose-response analysis (Fig. \ref{fig:4}d) shows that as $\mu_1$ increases from $0$ to $1$, the steady-state fraction of sensitive cells declines steadily, while resistant cells gradually become dominant when $\mu_1$ is greater than $0.53$. Notably, two key distinctions emerge when compared with the ODE model (Fig. \ref{fig:2}d):
\begin{enumerate}
\item[(1)] Even at low $\mu_1$, a small resistant subpopulation persists, suggesting that plasticity preserves latent resistance within the tumor, which can serve as a seed for future relapse.
\item[(2)] At higher $\mu_1$, the overall tumor burden decreases as drug efficacy increases. This behavior differs from the plateau observed in the ODE model. Phenotypic plasticity allows sensitive cells to temporarily transition into less proliferative, drug-tolerant states rather than being completely eliminated.   
\end{enumerate}

We next examined tumor evolution under adaptive therapy. Similar to the competitive ODE model, the total tumor burden oscillated periodically between $\alpha c_0$ and $c_0$ following treatment cycles (Fig. \ref{fig:5}a). However, in the plasticity-enabled model, the oscillations exhibit longer cycling periods and smoother transitions between treatment phases. This stabilization arises because phenotypic interconversion buffers population dynamics, preventing abrupt shifts in dominance between sensitive and resistant cells.

To quantitatively compare treatment performance, we calculated the two previously defined indices---average dosage and relative improvement---for the plasticity model. As shown in Fig. \ref{fig:5}b-c, the average dosage decreases as $\alpha$ increases, consistent with the ODE model, but the overall drug usage is systematically higher in the presence of plasticity. In contrast, the relative improvement (the ratio of the relative tumor burden reduction under AT to that under CT therapy) is lower when plasticity is included. Together, these results suggest that cell plasticity reduces the efficacy of adaptive therapy, necessitating higher cumulative drug exposure to maintain tumor control and resulting in smaller relative benefits compared to continuous therapy.

\begin{figure}[htbp]
\centering
\includegraphics[width=12cm]{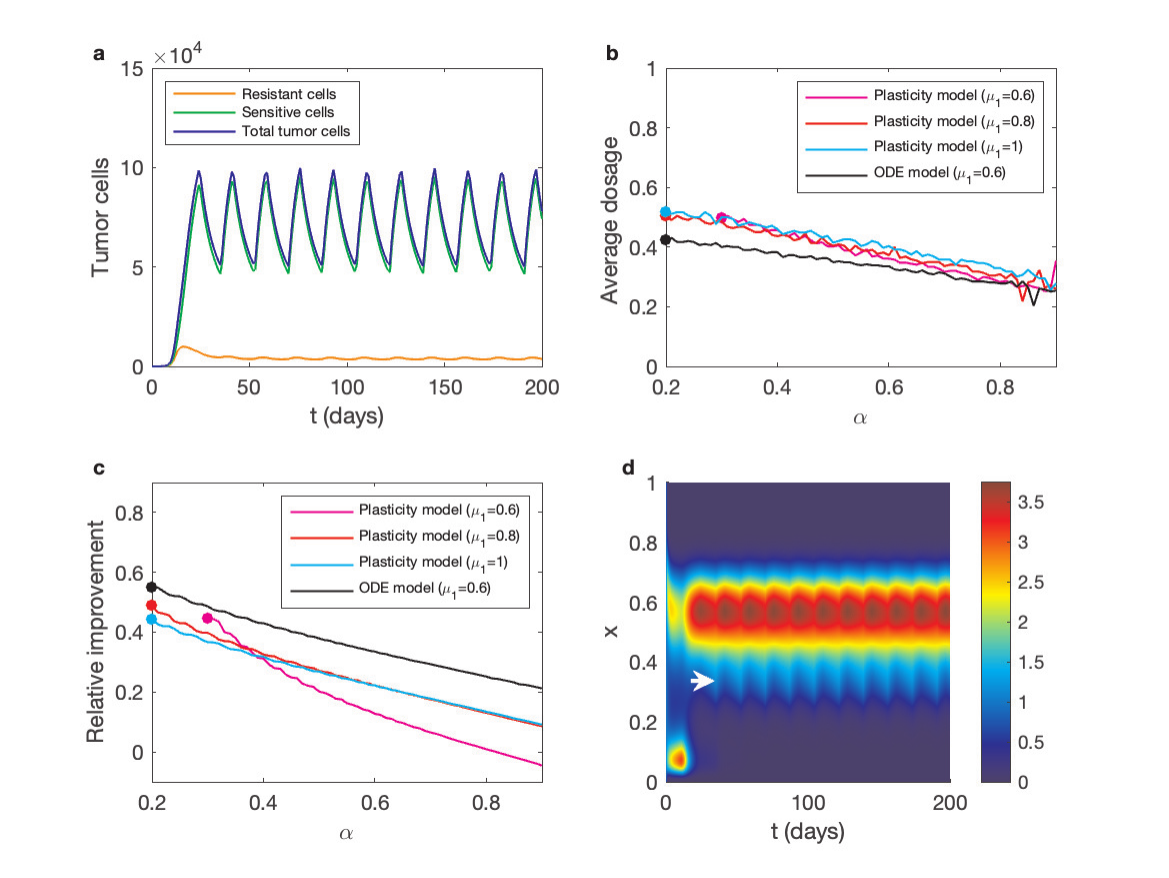}
\caption{Tumor dynamics under adaptive therapy in the plasticity-enabled model. (a) Temporal evolution of total tumor cells under adaptive therapy with $c_0 = 10^5$, $\alpha = 0.5$, and $\mu_1 = 0.6$. (b) Dependence of average dosage on the control parameter $\alpha$. (c) Dependence of relative improvement on $\alpha$. (d) Evolution of the sensitivity density function $f(t, x)$.}
\label{fig:5}
\end{figure}

To further illustrate the mechanistic basis of these differences, we examined the temporal evolution of the sensitivity density function
$$
f(t, x) = \dfrac{C(t, x)}{\int_0^1 C(t, x) \mathrm{d} x},
$$  
which represents the normalized distribution of drug susceptibilities across the tumor population (Fig. \ref{fig:5}d). During the early treatment phase, the fraction of highly sensitive cells ($x > 0.3$) declines sharply due to drug-induced killing. Concurrently, we observed a progressive shift of the distribution toward lower sensitivity values (indicated by the white arrow), reflecting the phenotypic drift from sensitive to tolerant states. Importantly, after treatment cessation, this distribution partially reverts toward higher sensitivity, confirming that drug tolerance in this model is reversible and dynamic rather than fixed.

Biologically, these findings suggest that tumor plasticity introduces a dynamic reservoir of transiently resistant cells that are continually replenished during therapy, thereby sustaining a residual tumor population even under aggressive dosing. This mechanism explains the observed quantitative differences in average dosage and relative improvement: plasticity demands more frequent or prolonged dosing to suppress regrowth, yet yields less overall therapeutic gain relative to CT. In contrast to purely genetic resistance captured by the ODE model, the plasticity-based framework highlights non-genetic adaptability as a fundamental barrier to complete tumor eradication under adaptive therapy. 

\subsection{Regulatory role of phenotypic plasticity in tumor evolution dynamics}
\label{sec3.3}

In the plasticity-enabled model \eqref{eq:4}, the probability inheritance function $p(x,y)$ governs the transmission of phenotypic traits during cell division and quantitatively describes the degree of tumor cell plasticity. The key parameter $\eta$ determines the width of the inheritance distribution: a larger $\eta$ corresponds to tighter phenotypic inheritance and lower plasticity, whereas a smaller $\eta$ produces broader variability and higher plasticity, enabling cancer cells to transition more readily between drug-sensitive and drug-resistant states. Biologically, high plasticity reflects epigenetic adaptability, transcriptional reprogramming, and stress-induced state switching---mechanisms widely implicated in the emergence of drug-tolerant persister cells and adaptive resistance.

To evaluate how phenotypic plasticity alters therapeutic outcomes, we varied $\eta$ and compared tumor progression under continuous therapy (CT) and adaptive therapy (AT).

\subsubsection*{Continuous therapy: plasticity accelerates relapse by supplying resistant phenotypes}

Under CT, decreasing $\eta$ (i.e., increasing plasticity) markedly accelerates relapse (Fig. \ref{fig:6}a). When plasticity is low ($\eta = 80$), therapy induces a steep initial decline in tumor burden, followed by an extended remission phase (day 50-100). During this window, the sensitivity density remains sharply concentrated (Fig. \ref{fig:6}b), indicating limited phenotypic switching and allowing a residual fraction of sensitive cells to maintain competitive suppression of resistant clones. Relapse occurs only when this balance is disrupted.

In contrast, when plasticity is high ($\eta = 40$ or $\eta = 60$), relapse occurs substantially earlier (Fig. \ref{fig:6}a). The sensitivity distribution exhibits significantly broader variance (Fig. \ref{fig:6}b), indicating rapid diversification of phenotypes and accelerated transitions from sensitive to resistant states. Drug pressure further amplifies these transitions by preferentially eliminating sensitive cells, allowing plasticity-generated resistant phenotypes to expand quickly.

These findings align with experimental studies showing that highly plastic tumors---such as melanoma \cite{Menon:2015hd}, glioblastoma \cite{Liau:2017aa}, and EGFR-mutant lung cancer \cite{Sharma:2009ge}---rapidly adopt drug-tolerant states under therapy. Even in the absence of pre-existing genetically resistant clones, dynamic reprogramming can sustain a persistent pool of therapy-evading cells, consistent with our simulation results.

\begin{figure}[htbp]
\centering
\includegraphics[width=12cm]{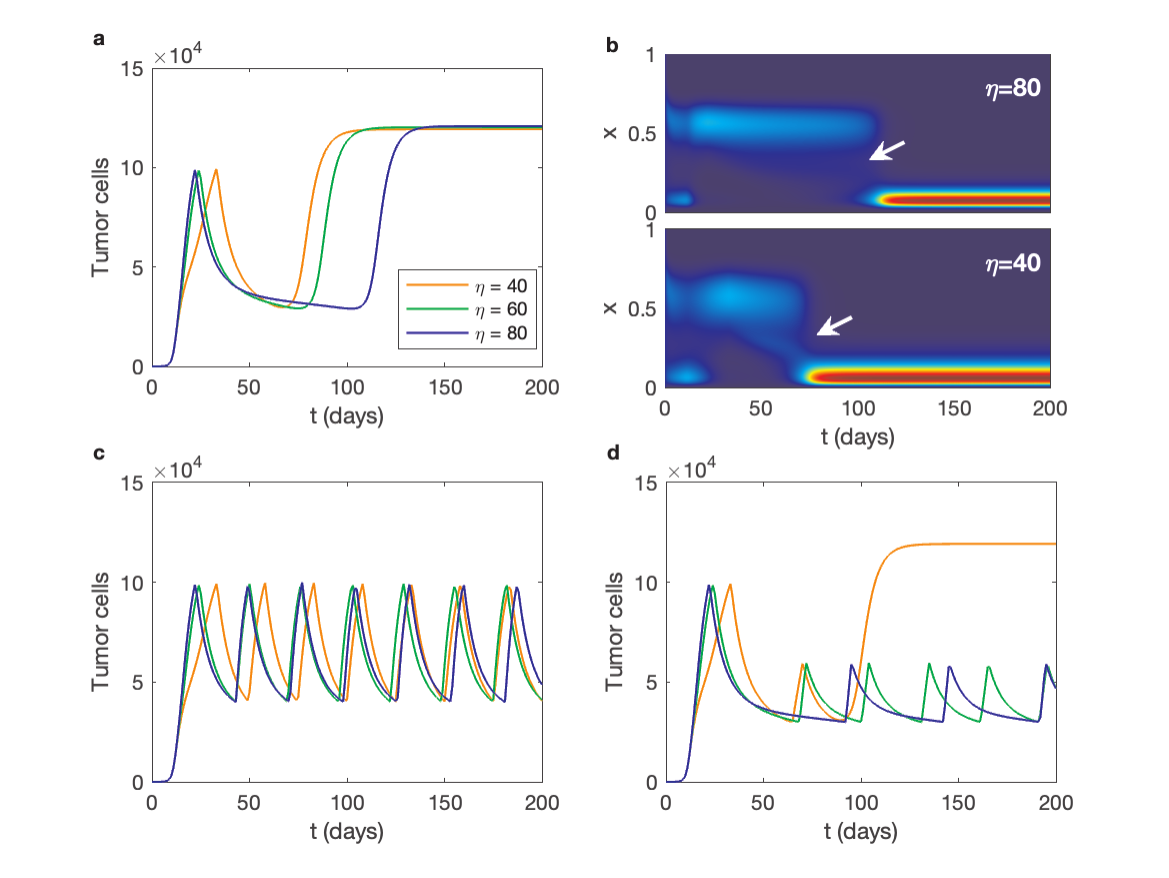}
\caption{Effects of phenotypic plasticity on tumor evolution dynamics. (a) Temporal evolution of total tumor cells under continuous therapy with $\eta = 40$, $60$, and $80$ ($\mu_1 = 0.6$). (b) Evolution of the sensitivity density function under continuous therapy for $\eta = 80$ and $\eta = 40$. (c) Temporal evolution of total tumor cells under adaptive therapy with $\eta = 40$, $60$, and $80$ ($\mu_1 = 0.6$, $\alpha = 0.4$). (d) Tumor dynamics under adaptive therapy with an early treatment protocol. The first treatment begins when $c(t) = 1.0 \times 10^5$ and pauses when $c(t) = 0.3 \times 10^5$. From the second cycle onward, treatment resumes once $c(t)$ rises to $0.60 \times 10^5$ and is suspended again when $c(t)$ falls to $0.3 \times 10^5$. This on-off cycle is repeated throughout the simulation. }
\label{fig:6}
\end{figure}

\subsubsection*{Adaptive therapy: Plasticity shortens treatment cycles while preserving overall control}

Under AT, phenotypic plasticity exerts similar regulatory effects (Fig. \ref{fig:6}c). High-plasticity tumors (lower $\eta$) exhibit more rapid tumor reductions during drug exposure and thus reach the treatment-off threshold earlier. Consequently, AT cycles become shorter, and the average dosage decreases.

Notably, however, the relative improvement of AT over CT is largely insensitive to $\eta$. This indicates that although plasticity modifies the temporal pattern of AT (cycle length, drug exposure duration), it does not fundamentally impair its capacity to maintain long-term tumor control. From a clinical perspective, these results suggest the following implications. High-plasticity tumors require shorter but more frequent AT cycles, reflecting their rapid phenotypic adaptation. Conversely, low-plasticity tumors tolerate longer treatment holidays, consistent with slower transitions towards resistant phenotypes. In both cases, AT effectively suppresses resistant expansion by leveraging competitive interactions with sensitive cells, provided treatment timing is properly managed.

\subsubsection*{Sensitivity to timing: Mistimed retreatment can result in resistance emergence}

A defining feature of AT is the dependence on sensitive-cell recovery during drug holidays. These cells exert suppressive pressure on resistant clones; insufficient recovery weakens this competition and increases the likelihood of rapid resistance expansion.

To test this vulnerability, we simulated an ``early restart'' regimen where therapy resumed prematurely once tumor burden reached $0.60\times 10^5$ (Fig. \ref{fig:6}d). Low-plasticity tumors ($\eta = 80$) remained relatively stable under this misaligned schedule, as limited phenotype switching preserved the sensitive-cell fraction and maintained competitive restraint. However, high-plasticity tumors ($\eta = 40,\ 60$) rapidly fail under the same conditions: early retreatment truncated sensitive-cell recovery, allowing resistant phenotypes---continuously generated through plastic switching---to accumulate over successive cycles, ultimately leading to therapy breakdown.

This result highlights a key limitation of AT: its success depends critically on allowing sufficient regrowth of sensitive cells. Intrinsic tumor plasticity narrows this therapeutic window, and suboptimal scheduling can inadvertently accelerate resistance evolution. These insights reinforce the clinical importance of personalized AT protocols and real-time monitoring to avoid mistimed dosing cycles.

Together, these results demonstrate that phenotypic plasticity profoundly shapes tumor evolutionary trajectories under both CT and AT. Plasticity accelerates relapse under CT, alters dosing requirements under AT, and can compromise treatment durability when cycle timing is suboptimal. Importantly, these effects imply substantial interpatient variability, as tumors with different plasticity levels respond differently to identical treatment protocols.

\subsection{Dynamic response to treatment by tumor cell phenotypic plasticity and population heterogeneity}
\label{sec3.4}

To quantify how inter-patient heterogeneity influences therapeutic outcomes, we constructed a large cohort of virtual patients by sampling key biological and treatment-related parameters across physiologically plausible ranges. To focus our discussions on heterogeneity in tumor growth and phenotypic plasticity, four tumor-intrinsic parameters were selected to represent variations in tumor growth, phenotypic flexibility, and drug susceptibility: the differentiation rate $\kappa$, the proliferation rate $\beta_0$, the phenotypic plasticity parameter $\eta$, and the drug-induced apoptosis rate $\mu_1$. For each virtual patient, these parameters were independently drawn from the following ranges:
$$
0.002 < \kappa < 0.007,\ 0.04 < \beta_0 < 0.12,\ 30 < \eta < 100,\ 0.5 < \mu_1 < 1.
$$ 
Other parameters are the same as those listed in Table \ref{table:1}. 

A total of $10,000$ parameter sets were generated, each representing one virtual patient. To mimic realistic clinical variability, the treatment-initiation time $t_0$ was also sampled uniformly from $20 < t_0 <30$. For each patient, therapy was initiated at the earliest of the sampled time $t_0$ or the moment when total tumor burden first exceeded $10^5$. The tumor size at treatment onset was recorded as $c_0$.

\subsubsection*{Heterogeneous outcomes under continuous therapy (CT)}

Simulation of all $10,000$ virtual patients revealed four qualitatively distinct tumor-evolution patterns under continuous therapy (Fig. \ref{fig:7}a):
\begin{enumerate}
\item[(1)] \textbf{Type I: Relapse with high-level resistance (62.6\%).}\ Tumor burden initially falls below $c_0/2$ but subsequently rebounds to levels exceeding the pre-treatment baseline. These cases represent the canonical pattern of initial remission followed by resistance-driven relapse, driven either by pre-existing resistance or plasticity-enabled phenotype switching.
\item[(2)] \textbf{Type II: Non-remission (19.5\%).}\ Tumors fail to decline below $c_0/2$ after therapy begins and eventually regrow beyond the pre-treatment level. High proliferative fitness or weak drug susceptibility prevents effective cytoreduction, rendering adaptive therapy infeasible.
\item[(3)] \textbf{Type III: Sustained control without relapse (12.8\%).}\ Tumors decrease after therapy onset and remain below baseline burden, representing cases where continuous therapy alone is sufficient to maintain long-term control.
\item[(4)] \textbf{Type IV: No-response (5.1\%).}\ Tumor burden does not decrease at all, suggesting intrinsic tolerance to the therapeutic regimen. 
\end{enumerate} 
To systematically define these categories, we introduced two quantitative indices: the short-term remission index $I_0$, defined as the minimum tumor burden normalized by $c_0$, and the long-term response index $I_1$, defined as the tumor burden at the end of simulation, also normalized by $c_0$. The classification scheme based on the $(I_0, I_1)$ plane is illustrated in Fig. \ref{fig:7}b.

\begin{figure}[htbp]
\centering
\includegraphics[width=12cm]{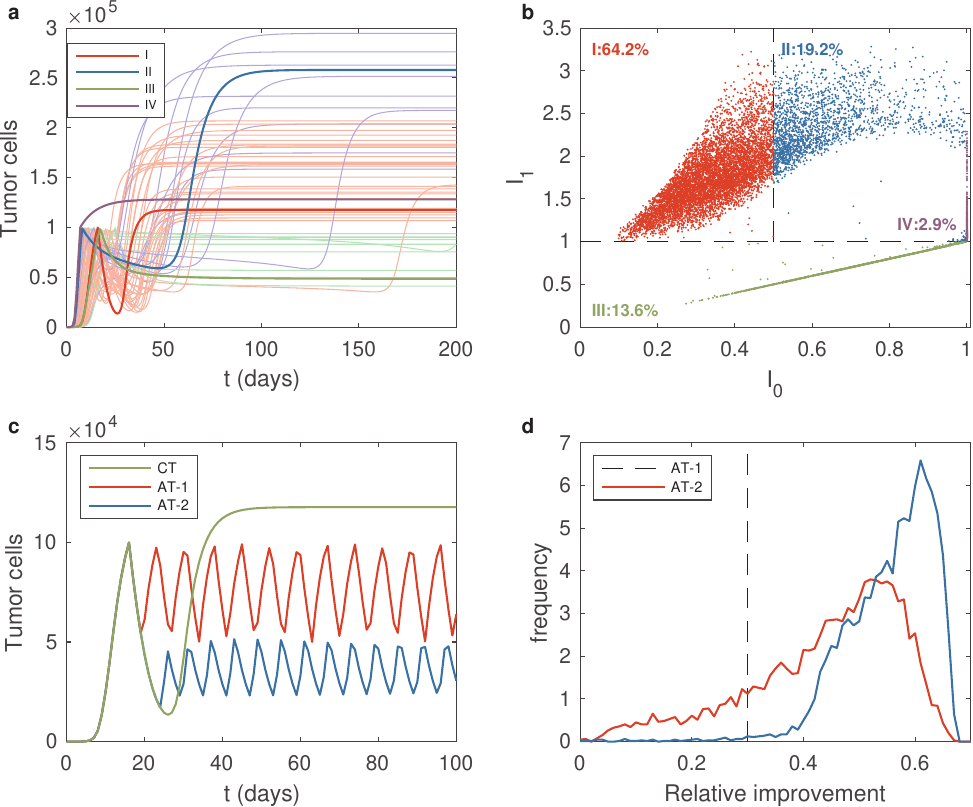}
\caption{Simulation of a virtual patient cohort reveals heterogeneous treatment responses. (a) Tumor evolution dynamics under continuous therapy (CT), with each color representing one of the four different types identified. (b) Classification of four tumor-evolution types using the indices $I_0$ (short-term remission) and $I_1$ (long-term response); numbers indicate percentages among $10,000$ virtual patients. (c) Representative tumor dynamics under the CT and adaptive therapy strategies AT-1 and AT-2 for the same virtual patient. (d) Distributions of relative improvement for adaptive therapy strategies AT-1 and AT-2 compared with continuous therapy. }
\label{fig:7}
\end{figure}

Type II, III, and IV patients exhibit limited clinical flexibility. Type II patients never reach the stopping threshold and are thus ineligible for adaptive therapy; type III patients already achieve durable control under CT, leaving no room for further improvement; type IV patients do not respond to therapy, indicating that some other regimens are required. Therefore, type I patients constitute the clinically relevant population for evaluating adaptive therapy because they exhibit both initial responsiveness and eventual relapse---the exact scenario for which AT was originally designed.

\subsubsection*{Performance of tumor-burden-based adaptive therapy (AT-1)}

For the clinically relevant Type I patients, we first evaluated a tumor-burden-based adaptive therapy (AT-1). Treatment was paused when tumor burden declined to $0.5 c_0$ and resumed once it regrew to $c_0$. This protocol is analogous to the original adaptive therapy proposed by Gatenby et al. \cite{Zhang:2022aa}, relying entirely on total tumor burden as the feedback signal.

Figure \ref{fig:7}c illustrates a representative comparison between AT-1 and CT. For each virtual patient, we computed the long-term relative improvement of AT-1 over CT, defined in \eqref{eq:ei}. The distribution of this improvement across all Type I patients is shown in Fig. \ref{fig:7}d.

Under AT-1, relative improvement ranges from $0$ to $0.67$, with a mean of $0.44$. Thus, most Type I patients benefit substantially from adaptive therapy. This improvement arises because treatment holidays maintain a pool of sensitive cells that competitively suppress resistant clones, delaying relapse and reducing overall tumor burden. However, the distribution also reveals significant heterogeneity: approximately $16.2\%$ of Type I patients exhibit improvement below 0.3, underscoring that a uniform AT strategy may be suboptimal for some patients.

\subsubsection*{Phenotype-guided adaptive therapy (AT-2) improves robustness across patients}

Building on the observation that phenotypic composition is a major driver of relapse, we evaluated a composition-based adaptive therapy (AT-2) after the initial treatment. The treatment rules were: stop therapy when resistant cells exceeded $20\%$ of the tumor population, and resume therapy when resistant cells fell below $10\%$. Figure \ref{fig:7}c shows a tumor evolution trajectory following the AT-2 strategy. This protocol requires real-time measurement of subpopulation composition---potentially achievable through tumor DNA, single-cell profiling, or phenotypic biomarkers.

Unlike AT-1, which manipulates resistance indirectly through the total burden, AT-2 directly targets resistance dynamics, intervening precisely when resistant cells begin to expand. Simulation results show clear benefits (Fig. \ref{fig:7}d). The overall distribution of relative improvement shifts upward, with the population-level mean increasing to $0.53$. Importantly, the fraction of patients with relative improvement below $0.3$ drops from $16.2\%$ (AT-1) to $1.2\%$ (AT-2); patients who benefited modestly under AT-1 show substantial gains under AT-2. Thus, AT-2 provides more consistent and superior outcomes across heterogeneous patient profiles, particularly by preventing resistant clones from crossing critical thresholds that lead to relapse.

\subsubsection*{Biological interpretation and implications}

Virtual patient analysis reveals the following key principles: phenotypic plasticity and proliferation heterogeneity can drive divergent outcomes under identical protocols. Tumor burden-guided therapy (AT-1) is beneficial but inherently limited, as its effectiveness depends on patient-specific, evolving tumor composition. In contrast, directly regulating the resistant fraction (AT-2) offers a stronger biological rationale—by intervening precisely when resistant populations rise, it aligns with evolutionary principles to prevent resistance dominance and minimize unnecessary drug exposure.

Together, these virtual-patient simulations underscore the importance of individualized treatment strategies that account for intrinsic tumor heterogeneity and phenotypic plasticity. They further demonstrate that phenotype-informed adaptive therapy offers more consistent and durable tumor control across heterogeneous patient populations, highlighting its potential to improve clinical outcomes beyond traditional burden-based approaches.

\subsection{Impact of cell-division time delay on tumor evolution dynamics}

In the mathematical formulation of the plasticity-enabled tumor model, cell proliferation was originally modeled as a delayed process to account for the finite duration of the cell cycle. Biologically, this delay reflects DNA replication, checkpoint regulation, and mitotic completion---processes that introduce an intrinsic time lag between commitment to division and the emergence of daughter cells. Although the full model incorporates this delay, previous simulations omitted it for simplicity, aligning with classical ODE-based models of tumor growth.

To evaluate whether ignoring this proliferation delay affects the qualitative or quantitative conclusions of our study, we reintroduced the proliferation delay and compared tumor evolution with and without delay across four scenarios: without therapy, continuous therapy (CT), tumor-burden-based adaptive therapy (AT-1), and phenotype-guided adaptive therapy (AT-2). Results are summarized in Fig. \ref{fig:8}.

\begin{figure}[htbp]
\centering
\includegraphics[width=12cm]{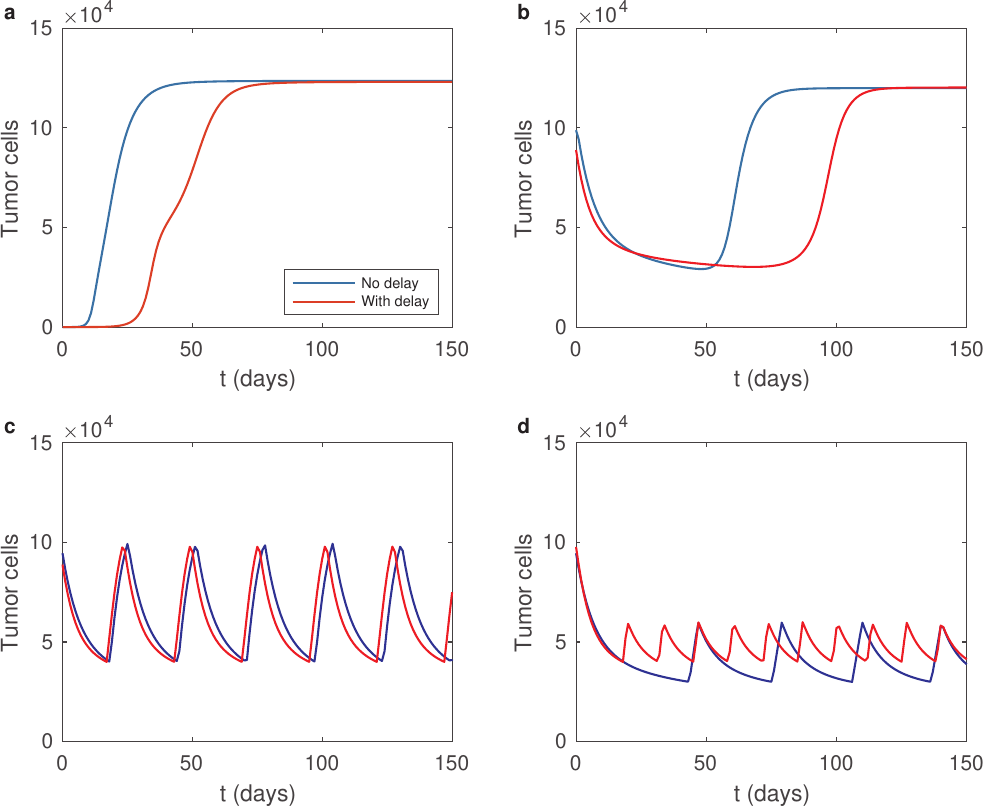}
\caption{Tumor dynamics with and without mitotic delay under different treatment strategies. (a) Without therapy. (b) Continuous therapy (CT).  (c) Tumor-burden-based adaptive therapy (AT-1). (d) Phenotype-guided adaptive therapy (AT-2). The simulations with delay were obtained by solving the integro-differential equation \eqref{eq:4}. We set $\tau = 24\ \mathrm{h}$. Accordingly, reset $\mu_0$ and $\mu_1$ as $\mu_0/24$ and $\mu_1/24$, respectively. Other parameters were the same as Fig. \ref{fig:6}. In (b)-(d), $t = 0$ corresponds to the time of starting the therapy.}
\label{fig:8}
\end{figure}

Figure \ref{fig:8}a compares uncontrolled tumor growth trajectories. Incorporating a proliferation delay slows early tumor expansion because cell division is delayed by the duration of the cell cycle. However, once the population enters a high-density regime, competition-driven dynamics dominate, and the two models converge to nearly identical long-term tumor burdens. Thus, while delay modifies transient growth kinetics, it does not affect the eventual tumor load.

Under continuous therapy (Fig. \ref{fig:8}b), both models exhibit an initial sharp decline in tumor burden---reflecting elimination of sensitive cells---followed by recurrence driven by the expansion of resistant phenotypes. Importantly, the presence of delay does not alter this qualitative behavior: relapse remains inevitable once resistant clones dominate. The primary effect is temporal---relapse occurs later when delay is included because resistant cells re-emerge more slowly. This extends the remission phase without altering the underlying evolutionary mechanism. 

For adaptive therapy strategies (AT-1 and AT-2), the tumor burden oscillates predictably between the predefined thresholds (Fig. \ref{fig:8}c-d). Introducing delay again preserves the qualitative structure of dynamics. The main difference is that treatment cycles may lengthen slightly, as delayed proliferation slows both tumor growth during holidays and sensitive-cell depletion during drug exposure. Nonetheless, the superiority of AT-2 over AT-1 and both over CT remains unchanged.

Overall, these results indicate that incorporating a biologically realistic proliferation delay only alters the timing of tumor dynamics, not the core mechanisms governing evolution, resistance, or treatment response. Thus, the no-delay formulation used in the main analysis captures the system's essential features, and the study's key conclusions remain robust.

\section{Discussion and Conclusions}
\label{sec4}

Adaptive therapy (AT), introduced by Gatenby et al. \cite{Gatenby_2009}, is an evolution-informed strategy aimed at controlling tumors rather than eradicating them. By maintaining drug-sensitive cells through dynamic dosing, AT reduces selective pressure that, under maximum tolerated dose (MTD), favors resistant clones \cite{Liu_2022, Zhang_2017, Kim_2021}. However, most mathematical models of AT rely on a binary sensitive-resistant framework and overlook phenotypic plasticity—the reversible switching between drug-sensitive and resistant states increasingly implicated in treatment failure. These binary ODE models cannot capture plasticity-driven resistance re-emergence or the persistence of low-level tolerant populations observed clinically.

In this work, we developed an integro-differential tumor evolution model representing drug sensitivity as a continuous phenotype, with inheritance governed by a probabilistic kernel $p(x, y)$. This formulation captures a continuum of phenotypes, enabling resistance to emerge through gradual state transitions rather than discrete clonal replacement. The framework integrates heterogeneity, plasticity, and therapeutic selection, bridging classical ODE models with more realistic evolutionary dynamics. Simulations reveal that plasticity maintains a persistent pool of low-sensitivity cells even in the absence of therapy, accounting for minimal residual disease. Under treatment, plasticity generates resistant phenotypes de novo, fundamentally altering evolutionary trajectories and making the sensitive-resistant balance a dynamic, reversible process shaped by treatment timing.

Our virtual patient analysis reveals that inter-patient variability in key parameters can lead to divergent outcomes under identical protocols. Tumor burden-guided adaptive therapy (AT-1) offers advantages but remains limited by its reliance on individualized, dynamic tumor composition. In contrast, phenotype-guided adaptive therapy (AT-2) provides a stronger biological rationale by directly intervening when resistant populations expand, aligning with evolutionary principles to prevent resistance dominance and reduce drug exposure. Virtual simulations demonstrate that personalized, phenotype-based strategies achieve more durable and stable tumor control across heterogeneous patient populations than traditional burden-guided approaches.

Despite its promising insights, this study has several limitations. First, while reintroducing the proliferation delay did not alter qualitative outcomes, the simplified no-delay model cannot capture scenarios in which cell-cycle regulation is dominant. Second, the framework currently considers monotherapy only; multi-drug adaptive protocols may further suppress resistance \cite{Asowed_2023}. Third, immune interactions—such as cytotoxic killing and immune exhaustion—were not included but represent a critical layer for future extensions. Finally, phenotype-guided adaptive therapy (AT-2) faces significant challenges in practical application, as existing technologies struggle to accurately distinguish resistant from sensitive clones, especially in cases of non-genetic resistance. As a result, tumor burden-guided strategies (AT-1) are generally more feasible. While emerging tools such as single-cell analysis and molecular imaging hold promise for real-time monitoring of resistance, they still require clinical validation.

Future research will expand this framework in several ways. First, incorporating temporal delays and cell cycle structures could enhance the accuracy of transient prediction. Second, while the current model assumes that the inheritance function $p(x,y)$ remains unchanged during therapy, future work should introduce treatment-intensity- or time-dependent mechanisms (e.g., drug-concentration functions) to better characterize the evolution of plasticity affected by treatment. Additionally, the proposed model should be extended to include multiple cancer cell types to assess cross-subpopulation evolution. Furthermore, integrating multi-drug therapies with tumor-immune interaction modules will help explore the synergistic effects between evolving therapies and immune modulation. Finally, achieving patient-specific calibration by correlating with multi-omics data (including single-cell sequencing, spatial transcriptomics, and epigenetic markers) will be crucial for advanced clinical translation.

In summary, this study identifies phenotypic plasticity as a key determinant of the efficacy of adaptive therapy. Our model, incorporating continuous sensitivity variation and phenotype inheritance, captures clinically relevant dynamics beyond those captured by classical ODE approaches. The findings demonstrate that adaptive therapy, when viewed through the lens of heterogeneity and plasticity, can transform biological variability from an obstacle into a controllable resource. This plasticity-enabled framework offers mechanistic insights and quantitative foundations for designing robust, personalized, evolution-informed treatment strategies in oncology.

\section*{Acknowledgments}

This work was funded by the National Natural Science Foundation of China (NSFC 12331018).

\bibliographystyle{unsrt} 
{\footnotesize \bibliography{myBib}}

\end{document}